\documentclass[12pt]{article}

\usepackage{amsmath}
\usepackage{amsfonts}
\usepackage{amssymb}
\usepackage{amsthm}

\theoremstyle{plain}

\newtheorem{thm}{Theorem}[section]
\newtheorem{lemma}[thm]{Lemma}
\newtheorem{prop}[thm]{Proposition}
\newtheorem{cor}[thm]{Corollary}

\theoremstyle{definition}

\begin{document}

\title{Harper operators, Fermi curves and Picard--Fuchs equations}

\author{
        \textsc{Dan Li} \\
        \\ 
Max Planck Institute for Mathematics\\
Vivatsgasse 7, 53111 Bonn, Germany\\
		Email: danli091981@gmail.com
        Picard--Fuchs equation, $q$-expansion
	    }

\date{}

\maketitle

\newpage

\begin{abstract}
This paper is a continuation of the work on the spectral problem of Harper operator using algebraic geometry.
We continue to discuss the local monodromy of algebraic Fermi curves based on Picard-Lefschetz formula.
The density of states over approximating components of  Fermi curves satisfies a Picard-Fuchs 
equation. By the property of Landen transformation, the density of states has a Lambert series as the quarter period. 
A  $q$-expansion of the energy level  can be derived from a mirror map as in the B-model. 
\end{abstract}

\section{Introduction}\label{Intro}
In solid state physics, we are interested in the behavior of moving electrons subject to periodic potential in a 2-dimensional lattice,
especially when  a uniform magnetic field is turned on.
We studied in \cite{L11} the spectral theory of Harper operator with irrational parameters 
 generalizing some results on Bloch variety and algebraic Fermi curves developed 
by Gieseker, Kn{\"o}rrer and Trubowitz \cite{GKT93}. 
Harper operator is the discrete magnetic Laplacian of a square lattice electron in a magnetic field, and its associated Bloch variety 
describes the complex energy-crystal momentum dispersion relation.
Harper operator arises in the study of integer quantum Hall effect, and it gives a special element in the noncommutative 2-torus \cite{BES94}.
By the famous ``Ten Martini Problem'' solved by Avila and Jitomirskaya \cite{AJ09}, the spectrum of the almost Mathieu operator
(i.e. 1-dimensional reduction of Harper operator) is a Cantor set of zero Lebesgue measure for all irrational magnetic fluxes, for 
a square lattice this Cantor structure is illustrated by Hofstadter butterfly. 

We consider the simplest model when the potential is zero under the independent electron approximation, 
and the Fermi surface could be roughly approximated
by real Fermi curves. In our case, each component of the Fermi curves gives rise to a stable family of elliptic curves with four singular fibers
belonging to the Beauville family \cite{Beauville82}. Peters and Stienstra \cite{PS89} also 
studied a similar pencil of K3-surfaces
related to the irrationality of $\zeta(3)$. By comparison, we are more interested in the density of states and 
the spectral aspect of Harper operator.

In \cite{L11}, we gave the geometric picture of the Bloch variety at an algebro-geometric level, then calculated the density of states
and some spectral function. 
In fact,  the compactification of the Bloch variety is an ind-pro-variety 
with properties analogous to a totally disconnected space, which is compatible with the Hofstadter butterfly picture. 
 The density of states over each 
component of Fermi curves can be calculated by an elliptic integral plus Landen transformations. 

In complex geometry,
a Picard-Fuchs equation (corresponds to a Gauss-Manin connection in the geometric setting) describes how the local period solutions 
change in a family. And the corresponding Schwarzian
equation characterizes  the change of the ratio of periods, i.e. the complex structure. Take into account the monodromy representation,
the exponential of the period ratio gives a local coordinate of the moduli space of complex structures.
In the context of mirror symmetry, Picard-Fuchs equation encodes the information from the B-model of the mirror partner, and Schwarzian
equation characterizes the mirror map, 
which is supposed to be a local identification between K{\"a}hler moduli and complex moduli in the large complex structure limit.

In this letter, we continue the investigation of the algebro-geometric properties of the Harper operator. 
We start with the monodromy of the Fermi curves,  the local monodromy will be computed 
in terms of Picard-Lefschetz formula, while the global monodromy was already discussed in \cite{GKT93}. 
Then we focus on a Picard-Fuchs equation satisfied by the density of states and the corresponding Schwarzian equation,
the density of states can be written as a hypergeometric function so that it has some modular property. 
Finally, from the relation between the elliptic modulus and the energy level,
we obtain  a $q$-expansion of the energy level based on a mirror map as in mirror symmetry.

\section{Fermi curves of Harper operator}  

In order to fix the notations, we briefly recall the spectral problem of Harper operator
discussed in \cite{L11}.
As a square lattice discrete magnetic Laplacian, Harper operator $H$ acting on $\ell^2(\mathbb{Z}^2)$ is defined by
\begin{equation}\label{Harper}
\begin{array}{rl}
H \psi (m,n)  := &   e^{-2\pi i \alpha n} \psi(m+1,n) +   e^{2\pi i \alpha n} \psi(m-1,n)  + \\
                &  e^{-2\pi i \beta m} \psi(m,n+1)  +     e^{2\pi i \beta m} \psi(m,n-1),
\end{array}                            
\end{equation}
where the two  magnetic translation operators
\begin{equation}
U\psi(m,n) : =e^{-2\pi i\alpha n} \psi(m+1,n) \ \  \text{ and }  \ \  V\psi(m,n) : =e^{-2\pi
i\beta m} \psi(m,n+1) 
\end{equation}
have irrational real parameters $\alpha$ and $\beta$ respectively. 
 
For distinct prime numbers $a$ and $b$, 
we consider the associated Bloch variety,
\begin{equation}\label{BHarp}
\begin{array}{rll} B := & \{ (\xi_1, \xi_2, \lambda)  \in   
\mathbb{C}^* \times \mathbb{C}^* \times \mathbb{C}
 \, | \, \exists \, \, \psi \neq 0 \, \, s.t. \,\, H \psi =\lambda \psi, \\  & \psi(m +a,n) =
\xi_1 \psi(m,n), \, \psi(m ,n +b) = \xi_2 \psi(m,n) \}.
\end{array}
\end{equation}
The Bloch variety is composed by all possible complex loci that can be
reached by analytic continuation of energy band functions with boundary conditions.

Due to the irrationality of the magnetic fluxes $\alpha$ and $\beta$, the loci determined by 
\eqref{BHarp} consists of a  countable collection of algebraic varieties. More precisely,
the Bloch variety in this case is an inductive limit of finite dimensional algebraic varieties,
we call it the Bloch ind-variety.

Using Fourier transformation, the Bloch ind-variety can be written as $B = \bigcup_{k,\ell, m,n} B_{k,\ell, m,n}$ 
with nonsingular subvarieties as components:
\begin{equation}\label{Bklmn}
B_{k,\ell,m,n} := \{ (\xi_1, \xi_2, \lambda) \in \mathbb{C}^* \times \mathbb{C}^* \times  \mathbb{C}~ | ~ N^{k,\ell}_{m,n}(\xi_1, \xi_2)
- \lambda = 0 \}
\end{equation}
where 
\begin{equation}\label{hatMkl}
 N^{k,\ell}_{m,n}(\xi_1,\xi_2):=  e^{ 2 \pi i\alpha (n+\ell b)} \xi_1  + e^{ -2 \pi i\alpha (n+\ell b)} \xi_1^{-1} +  
 e^{ 2 \pi i\beta (m+ka)} \xi_2 + e^{ -2 \pi i\beta (m+ka)} \xi_2^{-1} 
\end{equation}
for $(k, \ell, m,n) \in \mathbb{Z} \times \mathbb{Z} \times \mathbb{Z}_a \times  \mathbb{Z}_b $.

Now it is convenient to work on the approximating components of the Bloch ind-variety, and more details about its 
compactification by blow-ups can be found in \cite{L11}.


Consider the projection $\pi: B \rightarrow
\mathbb{C}$; $(\xi_1, \xi_2,\lambda) \mapsto \lambda$, if affine Fermi curves are defined by $F_\lambda(\mathbb{C}) :=
\pi^{-1}(\lambda)$, then each Fermi curve $ F_\lambda = \bigcup_{k,\ell, m,n}F_\lambda^{k,\ell,m,n}$ is an ind-variety with components 
\begin{equation}\label{Fklmn}
F_\lambda^{k,\ell,m,n} = \{ (\xi_1, \xi_2) |  N^{k,\ell}_{m,n}(\xi_1,\xi_2) = \lambda  \}.
\end{equation}
In \cite{PS89}, the authors used the parametrization $\lambda= \sigma + \sigma^{-1}$ based on the lattice structure in $H^2$.

For each $4$-tuple $(k,\ell,m,n)$, there exist singular fibers at $\lambda = 0, 4, -4$.  
It is easy to see that $F_0^{k,\ell,m,n}$ has two components:
\begin{equation}
 \begin{array}{ll}
 \{ (\xi_1, \xi_2) |  e^{ 2 \pi i\alpha (n+\ell b)} \xi_1 +  e^{ 2 \pi i\beta (m+ka)} \xi_2  = 0 \},\\
  \{ (\xi_1, \xi_2) | e^{ 2 \pi i\alpha (n+\ell b)} \xi_1 + e^{ -2 \pi i\beta (m+ka)} \xi_2^{-1} = 0  \}.
 \end{array}
\end{equation}
 $F_4^{k,\ell,m,n}$ is an irreducible curve with singularity only at  $
(e^{ -2 \pi i\alpha (n+\ell b)}, e^{ -2 \pi i\beta (m+ka)})$, similarly $F_{-4}^{k,\ell,m,n}$ is singular 
only at  $(- e^{ -2 \pi i\alpha (n+\ell b)}, - e^{ -2 \pi i\beta (m+ka)})$, and these singularities are ordinary double points.
Under the involution $(\xi_1, \xi_2) \leftrightarrow (- \xi_1, -\xi_2)$, it is natural to consider energy level $\mu = \lambda^2/16$ with 
singular fibers at $0, 1, \infty$, this
parametrization is related to the modular property of ${}_2F_1(\frac{1}{2}, \frac{1}{2}; 1; 16 \mu)$ as in \cite{R99}.

The projective closure of each component $F_\lambda^{k,\ell,m,n}$ in $\mathbb{P}^1 \times \mathbb{P}^1$, 
denoted by $\bar{F}_\lambda^{k,\ell,m,n}$, is an elliptic curve for a generic $\lambda$, 
and the complement $\bar{F}_\lambda^{k,\ell,m,n} \setminus F_\lambda^{k,\ell,m,n}$ is a divisor of type $(2,2)$. 
In other words, we have an elliptic fibration of the projective closure of each $B_{k,\ell,m,n}$, denoted by $\bar{B}_{k,\ell,m,n}$,  
\begin{equation}
\bar{\pi}: \bar{B}_{k,\ell,m,n} \rightarrow \mathbb{P}^1, \quad \bar{\pi}^{-1}(\lambda) = \bar{F}_\lambda^{k,\ell,m,n}
\end{equation} 
By the same statement as in \cite{GKT93}, the above fibration gives rise to a stable family of elliptic curves 
with four exceptional fibers, at $\lambda= \pm 4$ (type $I_1$), at $\lambda = 0$ (type $I_2$) and at $\lambda = \infty$ 
(type $I_8$). 
Hence the global monodromy group is $\Gamma_0(8) \cap \Gamma_0^0(4)$, and the details can be found in \cite{Beauville82}.


The local monodromy of the family of elliptic curves has been discussed in  \cite{GKT93}.
Here we have a similar result for the family of Fermi curve components in terms of Picard-Lefschetz formula. 
Let $\delta$ be a vanishing cycle,  recall that the Picard-Lefschetz transformation is given by 
\begin{equation}
T(x) = x - (x \cdot \delta) \delta 
\end{equation} 
for an arbitrary homology cycle $x$ and $x \cdot \delta$ is the intersection number between cycles with chosen orientation.

\begin{lemma} \label{lem1}
The local monodromies around $ \lambda = 4, 0, -4$ are given by the Picard-Lefschetz transformations $T_4, T_0, T_{-4}$,
\begin{equation}
T_4(\gamma) = \gamma + \delta_1,
  \quad
T_{-4}(\gamma) = \gamma - \delta_2,
 \quad
T_0(\gamma) = \gamma - \delta_1 + \delta_2,
\end{equation}
where $\delta_1 , \delta_2$ are vanishing cycles and $\gamma$ is an arbitrary homology cycle.
\end{lemma}
\begin{proof}
If we change the variables by  $\xi = e^{ 2 \pi i\alpha (n+\ell b)} \xi_1$ and $\eta = e^{ 2 \pi i\beta (m+ka)} \xi_2$, 
then the Fermi curve components become 
\begin{equation}
F_\lambda^{k,\ell,m,n} = \{ (e^{ -2 \pi i\alpha (n+\ell b)}\xi, e^{- 2 \pi i\beta (m+ka)}\eta) |  \xi + \xi^{-1}+ \eta + \eta^{-1}= \lambda  \}
\end{equation}
Rewrite the above curve as $\xi \eta^2 + (\xi^2 - \lambda \xi +1)\eta + \xi  = 0$ with discriminant $\Delta= (\xi^2-\lambda\xi +1)^2 - 4\xi^2$. Set $\Delta = 0$, there exist four branch points $\{ \frac{(\lambda +2)\pm \sqrt{\lambda^2 +4\lambda} }{2},  \frac{(\lambda -2)\pm \sqrt{\lambda^2 -4\lambda} }{2} \} $ on the $\xi$-plane. 

When $\lambda$ tends to $4$, $ \frac{(\lambda -2)\pm \sqrt{\lambda^2 -4\lambda} }{2}$ shrink to be an ordinary double point
 at $\xi = 1$, which corresponds to $(\xi_1, \xi_2) = (e^{ -2 \pi i\alpha (n+\ell b)}, e^{ -2 \pi i\beta (m+ka)})$. 
 And the cycle around  $ \frac{(\lambda -2)\pm \sqrt{\lambda^2 -4\lambda} }{2}$ is a vanishing cycle, call it $\delta_1$. Assume $\gamma$ is a cycle such that $\delta_1 \cdot \gamma = 1$, we rotate $\delta_1$ once and interchange the points $ \frac{(\lambda -2)+ \sqrt{\lambda^2 -4\lambda} }{2} \longleftrightarrow \frac{(\lambda -2) - \sqrt{\lambda^2 -4\lambda} }{2}$, by the Picard-Lefschetz transformation $T_4(\gamma)= \gamma + \delta_1 $.

When $\lambda$ tends to $-4$, $ \frac{(\lambda +2)\pm \sqrt{\lambda^2 +4\lambda} }{2}$ shrink to be an ordinary double point 
at $\xi = -1$, which corresponds to $(\xi_1, \xi_2) = (-e^{ -2 \pi i\alpha (n+\ell b)}, -e^{ -2 \pi i\beta (m+ka)})$.
And the cycle around  $ \frac{(\lambda +2)\pm \sqrt{\lambda^2 +4\lambda} }{2}$ is a vanishing cycle, call it $\delta_2$.
 $\gamma$ is the same cycle as above, so $\delta_2 \cdot \gamma = -1$, we rotate $\delta_2$ once and interchange the points 
 $ \frac{(\lambda +2)+ \sqrt{\lambda^2 +4\lambda} }{2} \longleftrightarrow \frac{(\lambda +2) - \sqrt{\lambda^2 +4\lambda} }{2}$, by the Picard-Lefschetz transformation $T_{-4}(\gamma)= \gamma - \delta_2$.

When $\lambda$ tends to $0$, $ \frac{(\lambda +2) \pm \sqrt{\lambda^2 +4\lambda}}{2}$ shrink to be one ordinary double point 
at $\xi = 1$ corresponding to $(\xi_1, \xi_2) = (e^{ -2 \pi i\alpha (n+\ell b)}, -e^{ -2 \pi i\beta (m+ka)})$, and 
at the same time $ \frac{(\lambda -2)\pm \sqrt{\lambda^2 -4\lambda} }{2}$ shrink to be the other ordinary double point 
at $\xi = -1$ corresponding to $(\xi_1, \xi_2) = (-e^{ -2 \pi i\alpha (n+\ell b)}, e^{ -2 \pi i\beta (m+ka)})$. 
The cycles  $\delta_1, \delta_2, \gamma$ are the same as above, so by the Picard-Lefschetz transformation 
$T_0(\gamma)= \gamma - \delta_1 + \delta_2 $.
\end{proof}

\section{Picard-Fuchs equation of density of states } \label{PFeq}

One of the main results of \cite{L11} is that  the density of states  on each component
of the Fermi curve can be identified with a period integral independent of the irrational parameters $\alpha, \beta$. In this section,
we will give the Picard-Fuchs equation of the density of states and relate the density of states to the modular subgroup $\Gamma(2)$. 

Recall that the density of states (DOS) on each component of the Fermi curve $F_\lambda^{k,\ell,m,n}$ 
 is of the form,
\begin{equation}\label{doseq}
DOS =  \frac{1}{2\pi^2 ab}(1+k) K(k) =  \frac{1 }{2\pi^2 ab} K(\frac{2\sqrt{k}}{1+k})
\end{equation}
where $k = \frac{1-\varepsilon}{1+\varepsilon}$ is the elliptic modulus.
 To avoid confusion with the elliptic modular lambda function, we write  the energy level as $\varepsilon = |\lambda|/4$ 
 instead of $\lambda$  used in the previous section.  
 The second equality in \eqref{doseq} is the ascending Landen transformation,
  which changes the corresponding period ratio from $\tau$ to $\tau/2$.

The complete elliptic integral of the first kind $K(k)$, also called the quarter period, satisfies a second order differential equation
\begin{equation}
k(1-k^2) \frac{d^2K}{dk^2} +(1-3k^2) \frac{dK}{dk} - k K = 0 
\end{equation}
with regular singularities at $k = -1, 0, 1, \infty$.
Indeed, if we change the variable by the elliptic modular lambda function $\lambda = k^2$ in $K$, then 
the above differential equation is equivalent to 
\begin{equation}
\lambda(1-\lambda)\frac{d^2K}{d\lambda^2} + (1- 2\lambda) \frac{dK}{d\lambda} - \frac{K}{4} = 0
\end{equation}
which is the Picard-Fuchs equation of the Legendre family of elliptic curves: $y^2 = x(x-1)(x- \lambda)$, as a special case of
hypergeometric differential equation. 

If we define 
\begin{equation}
D(k) := 2\pi^2 ab ~DOS = (1+k) K(k),
\end{equation}
then $D(k)$ satisfies a differential equation. Or equivalently,

\begin{prop}
The density of states  satisfies a second order differential equation 
\begin{equation}\label{DOSeq}
k(1-k)(1+k)^2\frac{d^2D}{dk^2} + (1-2k -k^2)(1+k)\frac{dD}{dk} + (k-1)D= 0 .
\end{equation}
From now on, we call it the Picard-Fuchs equation of  density of states.
\end{prop}
\begin{proof}
Recall the derivatives of $K(k)$ and  $E(k)$ (the complete elliptic integral of the second kind),
\begin{equation}
\frac{dK}{dk} = \frac{E(k)}{k(1-k^2)}-\frac{K(k)}{k}, \quad  
\frac{dE}{dk}=\frac{E(k)-K(k)}{k} 
\end{equation}
Substitute into the first derivative of $D(k)$,
\begin{equation}
\frac{dD}{dk}=(1+k)\frac{dK}{dk} + K(k) 
                      = \frac{E(k)}{k(1-k)} - \frac{K(k)}{k}
\end{equation}
In other words,
\begin{equation}\label{dD}
k(1-k)\frac{dD}{dk} = E(k)+ (k-1)K(k)
\end{equation}
Then
\begin{equation}
\frac{d}{dk}[k(1-k)\frac{dD}{dk}] = \frac{dE}{dk}+ (k-1)\frac{dK}{dk}+ K(k)
\end{equation}

\begin{equation}
k(1-k)\frac{d^2D}{dk^2} + (1-2k)\frac{dD}{dk} = \frac{E-K}{k}+ (k-1)(\frac{E}{k(1-k^2)}-\frac{K}{k}) +K =\frac{E}{1+k}
\end{equation}
Here we use (\ref{dD}) again to cancel $E(k)$, so 
\begin{equation}
k(1-k)\frac{d^2D}{dk^2} + (1-2k)\frac{dD}{dk} -\frac{k(1-k)}{1+k}\frac{dD}{dk} + \frac{k-1}{(1+k)^2}D = 0
\end{equation}

\end{proof}

Instead of the normal derivative $\frac{d}{dk}$, the logarithmic derivative operator $\Theta = k \frac{d}{dk}$ is widely used,
so the above Picard-Fuchs equation gives rise to a differential operator,
 \begin{equation}
   \Theta^2 +\frac{2k}{k^2 -1}\Theta - \frac{k}{(1+k)^2}.
 \end{equation}

\begin{cor}
  With respect to the energy level $\varepsilon$, the Picard-Fuchs equation of density of states is written as
  \begin{equation}
     \varepsilon(1-\varepsilon^2)\frac{d^2D}{d\varepsilon^2} + (1 - 3\varepsilon^2) \frac{dD}{d\varepsilon} - \varepsilon D = 0
  \end{equation}
 By setting $\vartheta = \varepsilon \frac{d}{d\varepsilon}$, it is equivalent to
  \begin{equation}
    [\vartheta^2 - \varepsilon^2(\vartheta + 1)^2] D = 0
  \end{equation}
\end{cor}

Thus the Picard-Fuchs equation of density of states is the same as that of the quarter period. Actually this is a natural consequence 
derived from Landen transformations. Let $L_k$ be the differential operator,
   \begin{equation}
     L_k = k(1-k^2) \frac{d^2}{dk^2} +(1-3k^2) \frac{d}{dk} - k 
   \end{equation}
   then the quarter period and its  complementary period give a fundamental set of independent solutions, i.e. $L_k K(k) = L_k K(k') = 0$, where $k^2 + k'^2 =1$.

\begin{prop}
  After Landen transformations, we still have 
  \begin{equation}
    L_{k_1'}K(k_1) = L_{k_1'}K(k_1') = 0
  \end{equation}
where 
\begin{equation}
  k_1 = \frac{2\sqrt{k}}{1+k}, \quad k_1' = \frac{1-k}{1+k}, \quad k_1^2 + k_1'^2 =1
\end{equation}
\end{prop}
With the convenient notations $K' = K(k') $ and $K_1 = K(k_1)$, one has $K_1 = (1+k) K(k)$, i.e. $K = ({(1+k_1')}/{2}) K_1$.
Look at the corresponding quotients,
\begin{equation}
  \frac{K'}{K} = 2 \frac{K_1'}{K_1}, \quad \text{since} ~~K' = (1+ k_1') K_1' = \frac{2K}{K'} K_1' 
\end{equation}
the half-period ratio $\tau = iK'/K$ changes as $\tau \mapsto \tau_1 = \tau/2$ according to the Landen transformation $k \mapsto k_1$.

By the classical results on the elliptic lambda function, we have $D = K_1 =  \tfrac{\pi}{2} \,{}_2F_1 \left(\tfrac{1}{2}, \tfrac{1}{2}; 1; 1- \varepsilon^2\right)$.
Similar to the nome $q = e^{i \pi \tau}$, define $q_1 = e^{i \pi \tau_1} = \sqrt{q}$, so the density of states has a Lambert series 
\begin{equation}
      DOS = \frac{1}{2 \pi^2 ab} [\frac{\pi}{2} + 2\pi\sum_{n=1}^\infty \frac{q_1^n}{1+q_1^{2n}}] = \frac{1}{4 \pi ab} [1 + 4\sum_{n=1}^\infty \frac{q^{n/2}}{1+q^{n}}].\, 
\end{equation}
And $\mu = \varepsilon^2 $ 
has a q-expansion about the infinite cusp

\begin{equation}
   \varepsilon^2(\tau) =1-16( q^{1/2}-8q+44q^{3/2}-192q^2+718q^{5/2}- \cdots ) 
\end{equation}
It is easy to see that
\begin{equation}
  1-\mu(4 \tau) = 1-\varepsilon^2(2 \tau ) = \lambda (2\tau) = k^2(\tau)
\end{equation}
i.e. $1-\varepsilon^2(\tau)$ is a Hauptmodul for the modular curve $X(2)$, $DOS(\varepsilon)$ is a weight-1 modular form  for $\Gamma(2)$.
The relation between $\varepsilon$ and $j$-invariant is then 
\begin{equation}
  j(\tau) =256 \frac{(\varepsilon^4 - \varepsilon^2 +1)^3}{\varepsilon^4(\varepsilon^2 -1)^2} (\tau) 
  = 16 \frac{(\varepsilon^4 -16\varepsilon^2 + 16)^3}{\varepsilon^8 (1-\varepsilon^2 )} (2 \tau)
\end{equation}

\section{mirror map and energy level}

In order to get a $q$-expansion of the energy level $\varepsilon$, we continue to look into 
the Picard-Fuchs equation of the density of states from a different point of view.
Since the monodromy representation for solutions of the Picard-Fuchs equation is the same as the geometric monodromy representation, we consider the local monodromy of $D(k) = K(2\sqrt{k}/(1+k))$, i.e. the monodromy of 
\begin{equation}
y^2 =(t^2 - 1)(\frac{4k}{(1+k)^2}t^2 -1) 
\end{equation}
The corresponding canonical holomorphic form is
\begin{equation}
\omega = \frac{dt}{y} = \frac{dt}{\sqrt{(t^2 - 1)(\frac{4k}{(1+k)^2}t^2 -1)} }
                      = \frac{ds}{\sqrt{(s^2 - 1)(s^2 - \frac{4k}{(1+k)^2} )} }
\end{equation}
where $ s=-1/t$. We have an equivalent statement as Lemma \ref{lem1} under the change of variable $k = \frac{1-\varepsilon}{1+\varepsilon}$.

\begin{lemma} \label{lem2}
The local monodromies around $ k = 0, 1$ are given by the Picard-Lefschetz transformations $S_0, S_1$,
\begin{equation}
S_0(\gamma) = \gamma + \delta_1,
  \quad
S_1(\gamma) = \gamma + 2\delta_2 -2\delta_3,
\end{equation}
\end{lemma}
\begin{proof}
If $k$ approaches $0$, then $\pm 2\sqrt{k}/(1+k)$ shrink to $0$, which gives an ordinary double point at $s= 0$. So the cycle around $\pm 2\sqrt{k}/(1+k)$ is a vanishing cycle, call it $\delta_1$. Let $\gamma$ be a cycle intersects with $\delta_1$ so that $\delta_1 \cdot \gamma = 1$. Hence by the Picard-Lefschetz formula $S_0(\gamma) = \gamma + \delta_1$.

If $k$ approaches $1$, then $\pm 2\sqrt{k}/(1+k)$ go to $\pm 1$, which gives two ordinary double points at $s= \pm 1$. So the cycles around $ 2\sqrt{k}/(1+k), ~ 1$ and $ -2\sqrt{k}/(1+k), ~ -1$ are vanishing cycles, call them $\delta_2, ~\delta_3$, let $\gamma$ be a cycle intersect with $\delta_2$ and $\delta_3$ such that $\delta_2 \cdot \gamma = 1$ and  $\delta_3 \cdot \gamma = -1$. Similarly, by the Picard-Lefschetz formula $S_1(\gamma) = \gamma + 2\delta_2 -2\delta_3$.
\end{proof}


As a multiple valued function, the local period solution
could change according to the monodromy  if it goes along any loop around some regular singularity.  
It is better to consider the ratio of the period solutions since its exponential is invariant subject to the local monodromy. 
In the following we change the Picard-Fuchs equation of density of states into its Q-form and consider the related Schwarzian equation.

\begin{lemma}
By change of variable $D = UV$, the Picard-Fuchs equation  \eqref{DOSeq} is equivalent to
\begin{equation}\label{Qform}
\frac{d^2U}{dk^2} + \frac{(1+k^2)^2}{4k^2(1-k^2)^2}U = 0
\end{equation} 
the corresponding differential operator is 
\begin{equation}
\Theta^2 + \frac{k^2}{(1-k^2)^2}
\end{equation}
\end{lemma}
\begin{proof}
Based on (\ref{DOSeq}), first define $V$ such that 
\begin{equation}
\frac{d}{dk}\ln V = \frac{k^2+2k -1}{2k(1-k^2)}
\end{equation}
and the second derivative of $V$ is
\begin{equation}
V'' = \frac{7k^2+2k - 3}{2k(1-k^2)}V' + \frac{1}{k(1-k)} V
\end{equation}
Then by substituting into the Picard-Fuchs equation, it is easy to get the Q-from with
\begin{equation}
Q(k) = \frac{(1+k^2)^2}{4k^2(1-k^2)^2}
\end{equation} 
\end{proof}

Recall that the Schwarzian derivative of a function $f(z)$ is defined by
\begin{equation}
 \{ f, z \}  = \left( \frac{f''(z)}{f'(z)}\right)'- \frac{1}{2} \left( \frac{f''(z)} {f'(z)} \right)^2,
\end{equation}
and it has a convenient inversion formula 
\begin{equation} \label{invfor}
\{ w, v\} = -\left(\frac{dw}{dv}\right)^2 \{ v, w\} .
\end{equation}

Suppose $D_1, D_2$ are two linearly independent local period solutions of the Picard-Fuchs equation (\ref{DOSeq}), 
and let $U_1, U_2$ be the corresponding solutions of the Q-from (\ref{Qform}), define $t = D_2/D_1 = U_2/U_1$, 
then $t$ satisfies the Schwarzian equation
\begin{equation}
\{ t, k \} = \frac{(1+k^2)^2}{2k^2(1-k^2)^2}
\end{equation}
By the inverse formula 
\begin{equation}
  \{ k, t \} = -k'(t)^2 \{ t, k \} = -\frac{k'^2(1+k^2)^2}{2k^2(1-k^2)^2}  
\end{equation}
Since the Schwarzian derivative 
is $GL(2, \mathbb{R})$ invariant, we have the following 

\begin{prop}
The energy level $\varepsilon$ satisfies a Schwarzian equation
\begin{equation}
  \{ \varepsilon, t \} =  -\frac{k'^2(1+k^2)^2}{2k^2(1-k^2)^2}  =  -\frac{\varepsilon'^2(1+\varepsilon^2)^2}{2\varepsilon^2(1-\varepsilon^2)^2}  
\end{equation}
\end{prop}


We now use the Frobenius 
method to find a set of local solutions,
write the Picard-Fuchs equation \eqref{DOSeq} as 
\begin{equation}
\frac{d^2D}{dk^2} + \frac{P_1}{k} \frac{dD}{dk} + \frac{P_2}{k^2} D= 0 
\end{equation}
with
\begin{equation}
P_1(k) = \frac{1-2k -k^2}{1-k^2}, \quad P_2(k) =  \frac{-k}{(1+k)^2}.
\end{equation}
 We want to find out two linearly independent solutions around $k=0$, which is a regular singular point. 
 First the indicial equation is given by 
\begin{equation}
r(r-1) + r P_1(0) + P_2(0) = r^2 =0.
\end{equation}
Thus $k=0$ has maximally unipotent monodromy, which agrees with Lemma \ref{lem2}. The indicial equation has repeated roots $r = 0$,
meaning that the local solutions around $k=0$ with normalization $D_1(0) = 1$ are given by
\begin{equation}
D_1 =  \sum_{n=0}^\infty d_n k^n, \quad D_2 =  D_1 ln \,k + \sum_{n=1}^\infty c_n k^n
\end{equation}

\begin{lemma}
By the Frobenius method, the coefficients are
\begin{equation} 
  \begin{array} {ll}
      d_{2n+1} =  d_{2n}=\left[ \frac{(2n)!}{2^{2n}(n!)^2}\right]^2 \\
     c_1 =0; ~~ c_{2n+1} = c_{2n} = 2 (\sum_{i=1}^{n} \frac{1}{i+n}) d_{2n}, ~n \geq 1 \\
   \end{array}
\end{equation}
\end{lemma}

Write out the local solutions explicitly,
\begin{equation}
   \begin{array}{ll}
     D_1 
          = (1 + k) ( 1+ \frac{1}{4}k^2  + \frac{9}{64} k^4  + \frac{25}{256} k^6  + \cdots ) \\
     D_2 =  D_1 ln \,k + (1+k)(\frac{1}{4}k^2 + \frac{21}{128} k^4  +\frac{185}{1536} k^6  + \cdots )
   \end{array}
\end{equation}
Hence the mirror map is 
\begin{equation}
  \begin{array}{ll}
   Q(k)= exp(D_2/D_1) = k  exp\{ \frac{2}{D_1} \sum_{n=1}^\infty \left[ \frac{(2n)!}{2^{2n}(n!)^2}\right]^2 
                (\sum_{i=1}^{n} \frac{1}{i+n}) k^{2n}   \} \\
  \quad ~~ \,= k + \frac{1}{4} k^3 + \frac{17}{128} k^5 + \frac{45}{512}k^7 +\cdots 
  \end{array}
\end{equation}

\begin{thm}
The energy level has a $Q$-expansion
\begin{equation}
\varepsilon   = 1-8 \left[ \frac{Q}{4} - 4 \left(\frac{Q}{4}\right)^2 +12\left(\frac{Q}{4}\right)^3-32 \left(\frac{Q}{4}\right)^4 
  +78 \left(\frac{Q}{4}\right)^5 - \cdots \right] 
\end{equation}
\end{thm}

\begin{proof}
Set $t = D_2 /D_1$ and  $Q= Q(k)= e^t$. 
The inverse of the mirror map is  
\begin{equation}
  k(Q) = Q -\frac{1}{4}Q^3 + \frac{7}{128} Q^5 - \frac{5}{512} Q^7 + \cdots
\end{equation}
Since the relation between $k$ and $\varepsilon$ is just  $\varepsilon = \frac{1 - k}{1 + k}$,  
\begin{equation} 
  \begin{array} {ll}
 \varepsilon 
   = 1-2Q + 2Q^2 -\frac{3}{2}Q^3 + Q^4 - \frac{39}{64} Q^5 + \frac{11}{64}Q^6 -  \cdots 
    \end{array}
\end{equation}
\end{proof}

\section{Discussion}

We use the classical method to study the Picard-Fuchs equation of density of states,  get a $q$-expansion 
of $\varepsilon^2(\tau)$ and a $Q$-expansion of $\varepsilon(t)$. Since the density of states is independent of 
the magnetic fluxes, we have to emphasize that these expansions are independent of these real parameters.

In the language of mirror symmetry, the imaginary variable $\tau \in \mathbb{H}$ is a local coordinate of the moduli space of complex 
structures, while $t$ is a local coordinate of the moduli space of K{\"a}hler structures. By the mirror hypothesis, one can identify
$\tau$ and $t$ in the large complex structure limit. In our case, we identify these local coordinates as $t = \tau/2$ and $Q = 4\sqrt{q}$.

It is easy to see that the terms in the $q$-expansion of the energy level $\varepsilon$ are all integers, 
which is also true in $ \varepsilon = 1+ \sum n_dq^{d/2}/(1- q^{d/2})$, i.e. $n_d \in \mathbb{Z}$.
Right now we don't have a 
good explanation of such ``instanton numbers'' related to the spectrum of Harper operator.

\nocite{*}
\bibliographystyle{plain}
\bibliography{dos}

\end{document}